\documentclass[11pt]{article}

\usepackage[letterpaper,margin=1in]{geometry}
\usepackage[utf8]{inputenc}
\usepackage[T1]{fontenc}
\usepackage{lmodern}
\usepackage{microtype}

\usepackage{amsmath,amssymb,amsfonts,amsthm,mathtools}

\usepackage{graphicx}
\usepackage{booktabs}

\usepackage{enumitem}
\usepackage{xcolor}
\usepackage{url}

\usepackage[hidelinks]{hyperref}
\usepackage[nameinlink,noabbrev]{cleveref}

\usepackage{algorithm}
\usepackage{algpseudocode}

\usepackage{mdframed}

\newtheorem{definition}{Definition}
\newtheorem{theorem}{Theorem}

\newtheorem{assumption}{Assumption}

\newcommand{\sysname}{\textsc{RTW-A}}
\newcommand{\sscgv}{\textsc{SSCGV}}
\newcommand{\srpo}{\textsc{SRPO}}
\newcommand{\llm}{\mathcal{L}}
\newcommand{\rUp}{R^{\uparrow}}
\newcommand{\rOpq}{R^{\circ}}
\newcommand{\wTau}{W^{\tau}}
\newcommand{\ahigh}{A^{H}}
\newcommand{\clean}{\ensuremath{\operatorname{Clean}}}
\newcommand{\external}{\ensuremath{\operatorname{External}}}
\newcommand{\tainted}{\ensuremath{\operatorname{Tainted}}}

\newcommand{\contam}{\ensuremath{\operatorname{Contaminated}}}
\newcommand{\deny}{\ensuremath{\operatorname{deny}}}
\newcommand{\allow}{\ensuremath{\operatorname{allow}}}
\newcommand{\guard}{\ensuremath{\operatorname{guard}}}
\newcommand{\highcap}{\ensuremath{\operatorname{HighCap}}}
\newcommand{\lease}{\ensuremath{\operatorname{Lease}}}
\newcommand{\promote}{\ensuremath{\operatorname{Promote}}}
\newcommand{\schema}{\ensuremath{\operatorname{Schema}}}
\newcommand{\sourceallowed}{\ensuremath{\operatorname{SourceAllowed}}}
\newcommand{\scopeallowed}{\ensuremath{\operatorname{ScopeAllowed}}}
\newcommand{\authority}{\ensuremath{\operatorname{Auth}}}
\newcommand{\ttl}{\ensuremath{\operatorname{TTL}}}
\newcommand{\untrusted}{\ensuremath{\operatorname{Untrusted}}}

\title{Autonomous LLM Agent Worms: Cross-Platform Propagation, Automated Discovery and Temporal Re-Entry Defense}

\author{
Mingming Zha\\
Indiana University Bloomington\\
\texttt{mzha@iu.edu}
\and
XiaoFeng Wang\\
Nanyang Technological University\\
\texttt{xiaofeng.wang@ntu.edu.sg}
}

\date{}

\begin{document}
\maketitle

\begin{abstract}
Autonomous LLM agents increasingly operate as long-running processes with persistent workspaces, memory files, scheduled task state, tool access, and messaging channel integrations. These features create a new class of propagation risk: attacker-influenced content can be written into persistent agent state, re-enter the LLM decision context through scheduled autoloading, and drive high-risk actions including memory updates, configuration changes, and cross-agent transmission. We present the first systematic framework for automated analysis of persistent worm propagation in file-backed multi-agent LLM ecosystems. Our work introduces \sscgv{}, an automated source-code graph vulnerability analyzer that traces data flow from file I/O to LLM context injection points and ranks injectable carriers by context injection position, requiring zero manual analysis or platform-specific knowledge. We additionally introduce \srpo{}, a summary-resilient payload optimizer that generates worm payloads robust to LLM-mediated summarization, paraphrasing, and compression across multi-hop agent communication. We evaluate on three open-source production agent frameworks and demonstrate zero-click autonomous propagation requiring no human interaction after initial injection, 3-hop cross-platform transmission across heterogeneous framework boundaries without platform-specific adaptation, inter-agent privilege escalation through trust-based delegation, and data exfiltration from agent workspaces. We additionally identify two empirical insights: user prompt carriers achieve substantially higher attack compliance than system prompt carriers, and read operations represent the primary integrity threat in LLM-mediated systems, inverting the traditional assumption that write access is the primary danger. To defend against this class of attacks, we develop \sysname{}, a temporal re-entry defense framework proven under a formal \emph{No Persistent Worm Propagation} theorem. The RTW constraint blocks write-before-exposed-read re-entry; sealed configuration protects static autoloaded files; typed memory promotion prevents free-form summaries from entering trusted memory; and capability attenuation limits high-risk actions after unavoidable external reads. Together these mechanisms eliminate the persistence, re-entry, action chain while preserving ordinary agent workflows. While coordinated disclosure is ongoing, we anonymize affected systems as Frameworks A--C and omit exploit-enabling details.
\end{abstract}

\section{Introduction}

The deployment of autonomous LLM agents in production environments has fundamentally altered the threat landscape for AI systems. Unlike traditional LLM applications that operate in isolated request-response cycles, contemporary agent platforms maintain persistent workspaces, load configuration files directly into their reasoning context, and communicate autonomously across shared channels. This architectural evolution creates novel attack surfaces where adversarial content can achieve persistent propagation, spreading autonomously between agents without human intervention.

Existing security research has demonstrated LLM worm attacks that achieve autonomous propagation within agent ecosystems~\cite{cohen2024morrisii, zhang2026clawworm}. However, these approaches rely on manually constructed attack vectors tailored to specific platforms and require detailed knowledge of target system architectures. Current worm implementations are platform-specific and cannot autonomously adapt to different agent frameworks, limiting their propagation scope to homogeneous environments. Moreover, existing attacks do not address the semantic degradation problem that occurs when payloads traverse realistic agent communication scenarios involving summarization, paraphrasing, and content compression. This paper introduces the first automated, cross-framework approach to LLM agent worm analysis, enabling systematic discovery of injectable surfaces across heterogeneous platforms and generation of summary-resilient payloads that maintain effectiveness across diverse agent architectures.

This paper introduces the first systematic framework for analyzing and defending against autonomous worm propagation in LLM agent ecosystems. Our investigation reveals a counterintuitive security primitive: in LLM-mediated systems, read operations can be more dangerous than write operations, as reading attacker-influenced content directly contaminates the decision-making process. This inverts traditional security intuitions where write access represents the primary threat vector.

The deployment model of LLM systems is shifting from stateless request-response interfaces to long-running autonomous agents. Modern agents maintain persistent workspaces, read and write local files, load natural-language bootstrap and memory files into their context, invoke external tools, and communicate through shared messaging channels. These architectural features are useful for autonomy and continuity, but they also create a new propagation substrate: natural-language content can persist in the environment, re-enter future LLM decision contexts, and influence subsequent tool use.

This paper studies file-backed worm propagation in autonomous LLM agent ecosystems. Unlike conventional software worms, the propagation mechanism does not require memory corruption, code injection, or exploitation of network protocol bugs. Instead, it exploits the semantic control loop of an LLM agent. An attacker places adversarial text in a location the agent will process, induces the agent to persist that text or a semantically equivalent derivative into an agent-controlled carrier, and relies on future exposed reads of that carrier to contaminate the LLM decision state. Once contaminated, the agent may perform high-risk actions within its normal permission scope, including updating memory, modifying configuration-like state, invoking tools, or communicating with other agents.

The central observation is that LLM agents introduce an integrity dimension to reads. In traditional access-control discussions, writes are often the primary integrity-critical operation, while reads are primarily associated with confidentiality. In LLM-mediated systems, an exposed read of attacker-controlled content can alter the controller that chooses future actions. In this sense, reading untrusted content can contaminate the decision maker.

We focus on the temporal pattern underlying persistent propagation:
\[
\wTau(f,t_w) \prec \rUp(f,t_r) \prec \ahigh(t_a),
\]
where \(\wTau\) is a tainted write to a persistent carrier \(f\), \(\rUp\) is an exposed read that reintroduces the carrier into an LLM or another authority-bearing interpreter, and \(\ahigh\) is a high-risk action. This pattern is more precise than static read/write permissions: a file is not dangerous merely because it is readable or writable; it becomes dangerous when attacker-influenced content is written and later exposed to a decision context that still has high-risk capabilities.

 
\noindent \textbf{Contributions.} In this manuscript, we make the following contributions.
\begin{itemize}[leftmargin=*]
 
\item \textbf{Automated injectable surface discovery (\sscgv{}).}
Given only the source repository of any agent platform, \sscgv{} automatically constructs a code property graph, traces data flow from file I/O to LLM context injection points (system prompt, user prompt, and memory loading), and ranks all injectable carriers by context injection position, requiring zero manual analysis or platform-specific knowledge. This enables systematic discovery across heterogeneous agent frameworks without per-platform customization.
 
\item \textbf{Summary-resilient payload optimization (\srpo{}).}
A three-LLM adversarial optimization system that automatically generates worm payloads robust to summarization, paraphrasing, and compression across four semantic dimensions (persist, propagate, harm, verbatim), directly addressing the semantic degradation bottleneck that renders existing text-based approaches ineffective in realistic agent communication pipelines.
 
\item \textbf{End-to-end worm attack against production agent frameworks.}
The first zero-click autonomous worm requiring no human interaction after initial injection, achieving persistent propagation across production-grade heterogeneous agent ecosystems. We demonstrate 3-hop cross-platform transmission chains spanning multiple production agent frameworks, inter-agent privilege escalation through trust-based delegation, and data exfiltration from agent workspaces. Neither zero-click operation nor cross-platform propagation across production frameworks without platform-specific adaptation is demonstrated by prior work.
 
\item \textbf{Two empirical security insights.}
First, context injection position determines carrier exploitability: user prompt carriers achieve substantially higher attack compliance than system prompt carriers because LLMs treat user-turn content as direct operational commands rather than background configuration. Second, in LLM-mediated systems \emph{read} operations can be more dangerous than \emph{write} operations, inverting the traditional security intuition that write access is the primary integrity threat; the dangerous operation is not writing the carrier but re-reading it into an authority-bearing context.
 
\item \textbf{Formal defense with end-to-end guarantees (\sysname{}).}
A \emph{No Persistent Worm Propagation} theorem (Theorem~\ref{thm:no-worm}) establishing that \sysname{} with RTW enforcement, persistent taint labeling, and capability attenuation prevents any attacker-controlled content from completing the propagation chain $\wTau \prec \rUp \prec \ahigh$ across any sequence of agents or carriers, providing mathematical guarantees against file-mediated worm propagation.
 
\end{itemize}

\paragraph{Disclosure note.}
This preprint intentionally anonymizes the evaluated frameworks and omits exploit-enabling details because coordinated disclosure is ongoing. We have notified the maintainers of all affected frameworks prior to public release of this preprint. We disclose the vulnerability class, the automated analysis methodology, aggregate results, and the formal defense, but withhold exact platform names, vulnerable carrier paths, payload templates, version identifiers, and exploitation scripts until mitigations are available or an agreed disclosure date is reached.

\section{Background and Threat Model}

\subsection{LLM Agent Architectures}

Modern personal-agent frameworks combine LLM reasoning with persistent local state, tool execution, and messaging-channel integrations. A typical deployment has a workspace directory that acts as the agent's working directory. The runtime may inject user-editable bootstrap files into the LLM context at session initialization, maintain long-term memory across sessions, and schedule proactive turns for background tasks or check-ins. Representative systems (e.g., OpenClaw-style personal-agent deployments) follow this pattern, with workspace files, bootstrap context injection, long-term memory, and scheduled operational turns.

We use the term \emph{persistent carrier} for any file-backed or state-backed object whose content can persist across turns or sessions and later be exposed to an LLM or another authority-bearing interpreter. Persistent carriers include bootstrap instruction files, identity or profile files, long-term memory, daily notes, scheduled task-state files, workflow files, scripts, cached summaries, and shared communication logs.

These carriers differ along four security-relevant dimensions. First, a carrier may be automatically read by the system, explicitly read in response to user or task requests, or read only by local tools. Second, a carrier may be writable by the agent, by an automatic summarization mechanism, by a trusted administrator, or by external synchronization. Third, the read sink matters: content exposed to an LLM, memory loader, configuration interpreter, RAG context, or execution engine has far greater authority than content read only by a local validator. Fourth, the temporal order matters: a write followed by a future exposed read creates a re-entry opportunity even if unrelated actions occur between them.

\subsection{Messaging Co-Presence and Trust Boundaries}

Agents often interact with users through shared messaging surfaces. In a single-operator personal assistant deployment (e.g., OpenClaw and NanoClaw), this design may be reasonable: messages are assumed to come from a trusted operator, and tool use is delegated to the agent within a personal trust boundary. In practice, however, agents may be placed in group channels, support multiple users, or observe messages not explicitly addressed to them. This weakens the single-operator assumption and creates an ambient co-presence setting where untrusted text can be processed by tool-enabled agents.

Our work studies this gap between the intended personal-assistant trust model and multi-party messaging deployments. We do not assume that the attacker can access the agent filesystem, modify platform source code, compromise model weights, or exploit traditional software vulnerabilities. The attacker can only place text into a channel, document, webpage, or workspace object that the agent may process during normal operation.

\subsection{Threat Model}

We consider multi-agent collaboration scenarios where autonomous LLM agents operate with read/write file system permissions and communicate through shared channels (Slack, Discord, Telegram). Agents exhibit heterogeneous privilege levels: low-privilege agents may be restricted to file I/O and messaging, while high-privilege agents possess capabilities such as arbitrary code execution. Trust relationships emerge naturally as agents process messages from other agents as legitimate operational context.

\textbf{Attacker Capabilities}: The attacker can send messages to public channels where target agents are present and place adversarial text into documents or content sources that agents may process. The attacker possesses knowledge of prompt engineering techniques and general agent framework patterns.

\textbf{Attacker Constraints}: The attacker cannot directly access agent file systems, modify source code, exploit traditional software vulnerabilities, or compromise model weights. Attack success depends entirely on inducing agents to perform authorized operations through natural language manipulation.

\textbf{Attack Objectives}: Achieve persistent autonomous propagation with inter-agent privilege escalation by (1) inducing agents to write adversarial content into persistent configuration files, (2) ensuring survival across restarts, (3) causing autonomous spread to additional agents, and (4) enabling privilege escalation through trust-based delegation.

\section{Attack Framework}
\subsection{Framework Overview}
 
Figure~\ref{fig:framework} presents our complete framework for analyzing and defending against persistent worm propagation in LLM agent ecosystems. The framework operates in three phases: automated vulnerability discovery, attack execution analysis, and formal defense construction.
 
\begin{figure}[t]
\centering
\includegraphics[width=\linewidth]{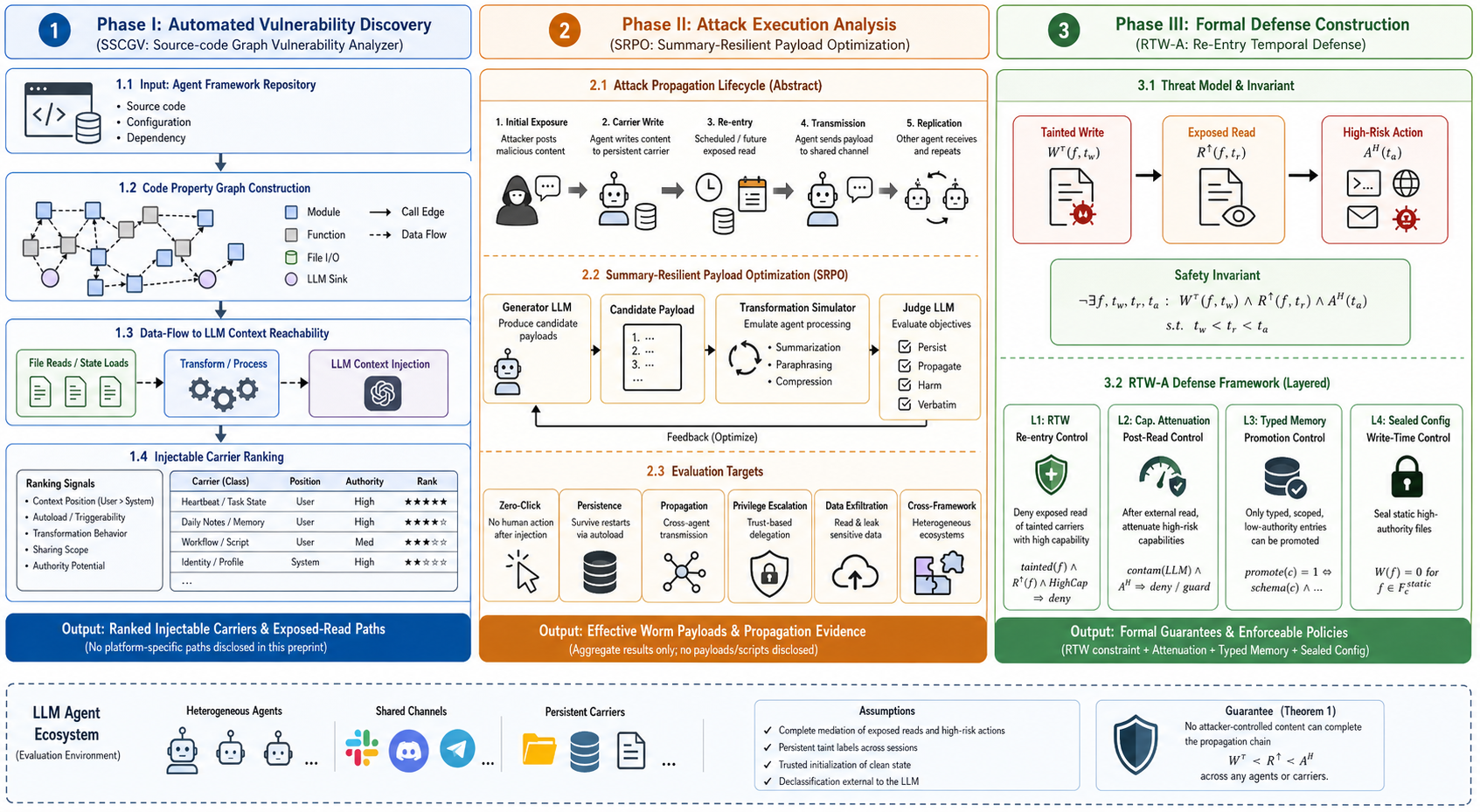}
\caption{
Framework overview for analyzing and defending against persistent worm propagation in LLM agent ecosystems. 
The framework operates in three phases: automated vulnerability discovery, attack execution analysis, 
and formal defense construction.
}
\label{fig:framework}
\end{figure}

\subsection{Static Source-Code Graph Vulnerability Analysis}

\sscgv{} automatically identifies injectable surfaces across heterogeneous agent platforms without manual platform-specific analysis. The input is a source repository. The output is a ranked set of persistent carriers, their exposed-read paths, and metadata describing context injection position, triggerability, transformation behavior, and propagation potential.

\paragraph{Phase 1: Code-property graph construction.} 
\sscgv{} constructs a code-property graph representing call relationships, control flow, data flow, filesystem operations, prompt construction, and tool invocation boundaries. The graph abstracts language-specific syntax into a common representation so that file-backed carrier flows can be compared across frameworks.

\paragraph{Phase 2: LLM-context reachability.} 
The analysis traces data flow from file reads, configuration loaders, memory loaders, scheduled task-state reads, template renderers, and retrieval components to LLM-context construction sites. The key target is not every file read but exposed reads: paths by which file content or an unbounded natural-language rendering can influence an LLM, memory loader, configuration interpreter, RAG context, execution engine, or agent-facing communication channel.

\paragraph{Phase 3: Carrier ranking.} 
\sscgv{} ranks candidate carriers primarily by \emph{context injection position}: whether the carrier's content is loaded into the user prompt or the system prompt, and where within the system prompt. User prompt carriers are ranked highest, as empirical evaluation shows they achieve substantially higher attack compliance than system prompt carriers (Section~\ref{sec:evaluation}). Secondary ranking signals include triggerability (whether the carrier is read automatically at session start or scheduled turns), transformation behavior (whether content is loaded verbatim or undergoes summarization and paraphrasing), and sharing scope (whether the carrier is accessible across multiple agents or sessions). We intentionally omit platform-specific carrier names, exact paths, ranking weights, and exploitation scripts from this preprint while disclosure is ongoing.

We intentionally omit platform-specific carrier names, exact paths, ranking weights, and exploitation scripts from this preprint while disclosure is ongoing.

\subsection{Summary-Resilient Payload Optimization}

Existing text-based propagation approaches often assume that payloads are transmitted verbatim between agents. Realistic agent communication violates this assumption: agents summarize conversations, compress memory, paraphrase task updates, and reformat messages before forwarding them. These transformations can weaken authority framing, drop procedural details, or remove replication instructions.

\srpo{} studies this semantic degradation problem through automated adversarial optimization. It treats payload generation as a multi-objective search problem over four dimensions:
\begin{itemize}[leftmargin=*]
    \item \textbf{Persistence:} the payload or its semantic derivative remains capable of inducing persistence after summarization or compression;
    \item \textbf{Propagation:} the payload induces an infected agent to transmit semantically effective content to another agent or shared channel;
    \item \textbf{Operational effect:} the payload preserves task-relevant objectives such as privilege delegation, tool use, or state modification under the experiment's controlled safety constraints;
    \item \textbf{Critical-token retention:} necessary identifiers, carrier references, or coordination markers survive transformation when such tokens are required for the experiment.
\end{itemize}

The optimization uses three LLM roles: a generator proposes candidates, a transformation simulator models realistic agent processing such as summarization and paraphrasing, and a judge evaluates whether the transformed output preserves the objectives. We report only aggregate robustness metrics in this preprint and omit prompt templates, payload instances, and judge rubrics while disclosure is ongoing.

\subsection{Propagation Mechanism}

The attack targets heartbeat-style operational state: persistent task-state carriers that the runtime periodically consults to decide whether proactive work should be performed. In some frameworks this role is implemented by an explicit heartbeat file; in others, \sscgv{} identifies the analogous scheduled task-state or operational-memory carrier.

This class of carrier is attractive because it combines persistence with scheduled exposed reads. Once attacker-influenced content is written into the carrier, a future scheduled turn can expose it to the LLM without additional attacker interaction. If the scheduled turn is configured to emit messages into a shared or last-contact channel, the contaminated decision state may generate output that reaches other agents. Those agents may then process the message, write a persistent derivative into their own carriers, and repeat the cycle.

The abstract propagation sequence is:
\begin{enumerate}[leftmargin=*]
    \item \textbf{Initial exposure:} the attacker places adversarial text into a channel or artifact processed by the target agent;
    \item \textbf{Carrier write:} the target agent writes attacker-influenced content or a semantic derivative into a persistent carrier;
    \item \textbf{Re-entry:} the carrier is later exposed to the LLM through an automatic or externally triggered read;
    \item \textbf{Transmission:} the contaminated decision state emits payload-bearing output to a shared communication surface;
    \item \textbf{Replication:} another agent observes the output and repeats the sequence.
\end{enumerate}

This propagation does not require the infected agent to explicitly target a particular peer. It exploits normal operational interactions and shared communication surfaces.

\section{Evaluation}
\label{sec:evaluation}

\paragraph{Disclosure status.}
The affected frameworks are anonymized as Frameworks A--C in this preprint. We have notified the maintainers of all affected frameworks prior to public release. We report aggregate results and carrier classes but withhold exact carrier paths, payloads, scripts, version identifiers, and platform names until coordinated disclosure is complete.

\paragraph{Evaluation overview.}
Our evaluation proceeds in eight stages that together characterize the full threat posed by autonomous LLM agent worms and validate our defense. We first establish the attack foundation: E1 identifies injectable carrier surfaces and shows that context injection position determines exploitability; E2 demonstrates that effective payloads can be generated to survive realistic agent communication pipelines. We then validate the attack end-to-end: E3 confirms that the complete propagation lifecycle succeeds across multiple heterogeneous frameworks; E4 quantifies propagation speed and scale. We next characterize the threat surface: E5 demonstrates concrete malicious outcomes including privilege escalation and data exfiltration; E6 shows that the attack is not model-specific and succeeds across tested LLM backends. Finally, we motivate and validate our defense: E7 shows that existing capability-based access control is insufficient because the attack operates entirely within authorized permissions; E8 evaluates the RTW-A defense framework and demonstrates that it blocks all demonstrated attack chains while preserving legitimate agent workflows.

\subsection{Experimental Setup}

We evaluated three open-source LLM agent frameworks that satisfy three criteria: persistent workspace state, exposed reads from file-backed or memory-backed carriers into LLM context, and messaging-channel integration. All experiments were conducted in controlled testbeds using accounts and agents operated by the authors. No third-party deployments, public users, or real user data were involved.

Each framework was deployed with a representative messaging-channel integration (Slack, Telegram) and an isolated workspace. Agents were configured with heterogeneous privilege levels: low-privilege agents restricted to file I/O and messaging, and high-privilege agents with arbitrary code execution capabilities. For each platform, \sscgv{} analyzed the source repository and produced ranked candidate carriers, which we then validated in controlled environments.

\subsection{E1: Injectable Surface Discovery and Context Position Effect}

\sscgv{} identified exposed persistent carriers in all three frameworks. We classify discovered carriers by their \emph{LLM context injection position}: whether carrier content is loaded into the system prompt at session initialization or into the user prompt during task execution. This classification reflects a security-relevant distinction beyond mere file presence.

\begin{table}[t]
\centering
\caption{Injectable carriers discovered by \sscgv{} per framework, classified by LLM context injection position. All carriers are persistent and writable within the agent's normal permission scope. Platform names and exact carrier paths are withheld during coordinated disclosure.}
\label{tab:sscgv-carriers}
\begin{tabular}{lccc}
\toprule
Framework & System Prompt Carriers & User Prompt Carriers & Total \\
\midrule
A & \textit{2} & \textit{9} & \textit{11} \\
B & \textit{2} & \textit{10} & \textit{12} \\
C & \textit{1} & \textit{7} & \textit{8} \\
\midrule
Total & \textit{5} & \textit{26} & \textit{31} \\
\bottomrule
\end{tabular}
\end{table}

\paragraph{Context position effect.}
We evaluated attack payloads delivered through carriers at both injection positions and observed a consistent compliance rate difference across all three frameworks. \Cref{tab:position-effect} summarizes the observed compliance rates.

\begin{table}[t]
\centering
\caption{Observed attack compliance rate by LLM context injection position, evaluated across discovered carriers in all three frameworks. Detailed per-carrier quantitative measurements are deferred to a subsequent version of this work.}
\label{tab:position-effect}
\begin{tabular}{lcc}
\toprule
Injection Position & Representative Carriers & Observed Compliance \\
\midrule
User prompt           & Heartbeat, task-state   & High   \\
System prompt         & Identity, memory, policy & Lower  \\
\bottomrule
\end{tabular}
\end{table}

User prompt carriers consistently achieve higher attack compliance than system prompt carriers across all three frameworks. We attribute this to two complementary mechanisms. First, instruction hierarchy training causes models to treat user-turn content as direct operational commands rather than background configuration~\cite{wallace2024instruction}. Second, user prompt content is injected immediately before the model response, exploiting recency bias~\cite{liu2024lost}. Among system prompt carriers, we also observe that injection position within the system prompt correlates with compliance rate, though a detailed quantitative analysis of within-system-prompt position effects is deferred to a subsequent version of this work.

\begin{mdframed}[
  linewidth=0.8pt,
  backgroundcolor=gray!8,
  innerleftmargin=6pt,
  innerrightmargin=6pt,
  innertopmargin=4pt,
  innerbottommargin=4pt
]
\textbf{Insight 1: Context Position Determines Carrier Exploitability}
The exploitability of a persistent carrier is determined not only by its writability or autoload frequency, but by \emph{where} its content is injected into the LLM context. User prompt carriers achieve substantially higher attack compliance than system prompt carriers because LLMs treat user-turn content as direct operational instructions. Heartbeat-style task-state files, loaded into the user prompt during scheduled turns, represent a disproportionate risk relative to identity or policy files loaded at session initialization.

\end{mdframed}

\subsection{E2: Summary-Resilient Payload Optimization}

We evaluated \srpo{} by passing candidate payloads through controlled transformation pipelines that emulate summarization, paraphrasing, and compression. We measure semantic retention along the four objectives described above: persistence, propagation, operational effect, and critical-token retention. To avoid publishing exploit-enabling strings, this preprint reports the optimization setup and aggregate trends but omits payload examples and judge prompts.

Across the evaluated transformations, \srpo{}-optimized payloads increased retention of persistence and propagation intent relative to non-optimized baselines. Verbatim retention remained sensitive to compression and paraphrasing strength, motivating the multi-objective optimization approach. These results confirm that realistic agent communication does not inherently prevent worm propagation when payloads are designed to be semantically resilient.

\subsection{E3: End-to-End Propagation Validation}

\subsubsection{E3.1: Per-Framework Propagation}

We validated end-to-end propagation in isolated testbeds for all three frameworks. Starting from a single adversarial message sent by an unprivileged test account controlled by the authors, the attack completed the full abstract lifecycle: initial exposure, persistent carrier write, future re-entry, scheduled or interaction-driven transmission, and infection of a subsequent agent. \Cref{tab:e2e-results} summarizes the results across frameworks.

\begin{table}[t]
\centering
\caption{End-to-end propagation validation results per framework. All attacks were conducted in isolated controlled testbeds by the authors. Platform names are withheld during coordinated disclosure.}
\label{tab:e2e-results}
\begin{tabular}{lcccc}
\toprule
Framework & Persistence & Re-entry & Propagation & Privilege Esc. \\
\midrule
A & \checkmark & \checkmark & \checkmark & \checkmark \\
B & \checkmark & \checkmark & \checkmark & \checkmark \\
C & \checkmark & \checkmark & \checkmark & \checkmark \\
\bottomrule
\end{tabular}
\end{table}

In all three frameworks, the attack required zero human interaction after the initial message injection. Propagation was triggered autonomously by the scheduled heartbeat mechanism, confirming the zero-click nature of the attack.

\subsubsection{E3.2: Cross-Framework Propagation}

We also validated cross-framework propagation in which an infected agent on one framework transmitted a payload-bearing message to an agent running on a different framework through a shared messaging channel. The receiving agent, operating under its own runtime and workspace configuration, successfully processed the payload, wrote it into its own persistent carrier, and completed the propagation lifecycle autonomously.

This demonstrates that worm propagation is not limited to homogeneous agent ecosystems. An attacker who infects a single agent can achieve lateral spread across heterogeneous platforms without any platform-specific adaptation, provided the agents share a common messaging surface.

\begin{mdframed}[
  linewidth=0.8pt,
  backgroundcolor=gray!8,
  innerleftmargin=6pt,
  innerrightmargin=6pt,
  innertopmargin=4pt,
  innerbottommargin=4pt
]
\textbf{Insight 2: In LLM Agents, Read Is More Dangerous Than Write}

In traditional systems, \emph{write} access is the primary integrity threat. In LLM-mediated agents, \emph{read} can be more dangerous: an exposed read of attacker-influenced content directly contaminates the LLM decision state, inducing the agent to execute high-risk actions within its normal permission scope. The dangerous operation is not writing the carrier but re-reading it into an authority-bearing context; the temporal pattern $W^\tau(f) \prec R^\uparrow(f) \prec A^H$ is the true attack primitive.
\end{mdframed}

\subsection{E4: Propagation Efficiency}

We measure the speed and scale of autonomous worm propagation to quantify the practical threat. \Cref{tab:propagation-efficiency} reports propagation metrics across frameworks and the cross-framework chain.

\begin{table}[t]
\centering
\caption{Propagation efficiency across frameworks. Hops denotes the number of agents infected in a single chain. Time denotes wall-clock time from initial injection to completion. All attacks required zero human interaction after initial message delivery.}
\label{tab:propagation-efficiency}
\begin{tabular}{lcccc}
\toprule
Scenario & Hops & Human Interaction & Zero-click \\
\midrule
Framework A        & \textit{3} &  None & \checkmark \\
Framework B        & \textit{3} & None & \checkmark \\
Framework C        & \textit{3} &  None & \checkmark \\
Cross-framework    & 3          & None & \checkmark \\
\bottomrule
\end{tabular}
\end{table}

Propagation speed is bounded primarily by the heartbeat interval: once a carrier is infected, the next scheduled turn autonomously transmits the payload to all reachable agents on shared messaging surfaces. This means that in deployments with short heartbeat intervals, the entire reachable agent ecosystem can be compromised within minutes of the initial injection, without any further attacker interaction.

\subsection{E5: Malicious Payload Taxonomy}

Beyond autonomous propagation, we evaluate the range of malicious behaviors achievable through the worm mechanism. We identify two primary attack classes demonstrated in our controlled evaluation.

\paragraph{E5.1: Privilege Escalation via Trust-Based Delegation.}
We evaluated cross-privilege propagation in which a low-capability infected agent transmitted payload-bearing content to a higher-capability agent through a shared messaging surface. The higher-capability agent, which possessed shell execution and network access capabilities, processed the content under its own elevated tool permissions and executed attacker-directed actions. This demonstrates lateral movement through trust-based delegation: the attacker does not need to compromise the high-privilege agent directly, as the worm propagates through normal inter-agent communication channels by exploiting the implicit trust agents place in peer messages.

\paragraph{E5.2: Data Exfiltration.}
We evaluated payloads designed to exfiltrate sensitive information from the agent workspace. An infected agent was induced to read its own system prompt, identity configuration, and workspace files, then transmit the contents through an external messaging channel. This demonstrates that the worm mechanism can be used not only for propagation but also for targeted information extraction, including operator-level configuration and user data present in the agent workspace.

\subsection{E6: Model Robustness}

We evaluate whether the attack is specific to a particular LLM backend or generalizes across models. \Cref{tab:model-robustness} reports attack success rates across the LLM backends tested in our evaluation.

\begin{table}[t]
\centering
\caption{Attack compliance rate across LLM backends. All models were evaluated under identical experimental conditions using the same carrier and payload configuration.}
\label{tab:model-robustness}
\begin{tabular}{lcc}
\toprule
Model & Single-hop Compliance & Multi-hop Compliance \\
\midrule
GPT-4o-mini      & \textit{100}\% & \textit{100}\% \\
Gemini-2.5-Flash & \textit{100}\% & \textit{100}\% \\
\bottomrule
\end{tabular}
\end{table}

The attack succeeds across all tested model backends. This suggests that file-mediated worm propagation is primarily a structural property of the agent architecture rather than a model-specific weakness: the attack exploits the agent's file system permissions and scheduled autoloading mechanism, operating at the architecture level rather than targeting any particular model's instruction-following behavior. We leave a comprehensive multi-model evaluation to a subsequent version of this work.

\subsection{E7: Capability--Risk Tradeoff and Limitations of Existing Access Control}

We evaluate whether existing capability-based access control mechanisms are sufficient to prevent worm propagation. Using Framework A's built-in permission system as a representative deployment, we test attack success under different permission configurations to characterize the relationship between agent capabilities and attack impact.

\begin{table}[t]
\centering
\caption{Attack success under different permission configurations in Framework A. Even restrictive configurations remain partially vulnerable because the attack operates entirely within authorized capabilities.}
\label{tab:capability-risk}
\begin{tabular}{lcccc}
\toprule
Permission Config & File Write & Messaging & Persistence & Propagation \\
\midrule
Full permissions      & \checkmark & \checkmark & \checkmark & \checkmark \\
Messaging disabled    & \checkmark & \texttimes  & \checkmark & \texttimes  \\
File write disabled   & \texttimes  & \checkmark & \texttimes  & \checkmark \\
Minimal permissions   & \texttimes  & \texttimes  & \texttimes  & \texttimes  \\
\bottomrule
\end{tabular}
\end{table}

Existing capability-based access control reduces the attack surface but does not eliminate the threat. With messaging disabled, the worm achieves persistent infection but cannot propagate. With file write disabled, propagation remains possible but persistence is prevented. Only the minimal permission configuration, which also prevents legitimate agent operation, blocks both attack phases.

The fundamental limitation of capability-based control is that it cannot distinguish between legitimate and attacker-induced uses of the same capability. A write to a carrier file and a read from that file are both authorized operations. The dangerous pattern is the \emph{temporal sequence} between them: a tainted write followed by an exposed read that re-introduces attacker-influenced content into the LLM decision context. Capability-level access control has no mechanism to enforce constraints on this temporal ordering.

This result directly motivates the RTW defense framework: the missing security primitive is not capability restriction but \emph{temporal re-entry control}, which operates below the capability layer and enforces constraints on the sequence $W^\tau(f) \prec R^\uparrow(f) \prec A^H$ independently of what capabilities the agent holds.

\subsection{E8: Defense Evaluation}

We evaluated the RTW-A defense framework against the demonstrated attack chains in the same controlled testbeds. The defense was applied to Frameworks A without modifying agent functionality or disabling legitimate operations.

\paragraph{RTW enforcement.} RTW prevented same-carrier write-before-exposed-read re-entry for ordinary workspace carriers. When a tainted carrier was scheduled for autoload, the runtime denied the exposed read before the content reached the LLM decision context, blocking the re-entry step of the attack chain.

\paragraph{Static configuration sealing.} Sealed configuration prevented runtime modification of high-authority bootstrap files. Attempts by contaminated LLM decision states to overwrite identity or policy carriers were denied at write time, independent of the RTW constraint.

\paragraph{Typed memory promotion.} Typed memory promotion prevented free-form summaries and LLM-generated content from becoming trusted autoloaded memory. Worm payloads embedded in memory candidates were rejected by the schema and injection-pattern checks in the promotion gate, preventing persistence through the memory channel.

\paragraph{Capability attenuation.} Capability attenuation prevented contaminated LLM decision states from invoking high-risk actions after unavoidable external reads. Attempts to forward worm payloads via messaging channels or write to persistent carriers following external message ingestion were denied by the runtime policy engine.

Together, these mechanisms blocked all demonstrated attack chains without disrupting legitimate agent workflows. Agents continued to process user requests, read workspace files, update memory, and send messages in response to authorized operations; only operations that completed the temporal re-entry chain or followed external contamination were denied.

We defer release of platform-specific defense patches, policy tables, and latency overhead measurements until coordinated disclosure is complete. The formal security properties of the defense are established in \Cref{sec:defense}.
\section{Defense Framework}
\label{sec:defense}

\subsection{Overview: Temporal Re-Entry Control}

We propose \sysname{}, a defense-in-depth framework for persistent prompt-propagation attacks in LLM agents. The key observation is that file-mediated compromise is not caused by file access alone, but by temporal re-entry: attacker-influenced content is written into a persistent carrier and later reintroduced into an authority-bearing decision context.

At a high level, the attack chain has the following structure:
\[
\wTau(f,t_w)
\;\rightarrow\;
\rUp(f,t_r)
\;\rightarrow\;
\contam(\llm,t_r)
\;\rightarrow\;
\ahigh(t_a),
\]
where \(f\) is a persistent carrier, \(\rUp\) denotes an exposed read that reintroduces file content into the LLM or another authority-bearing interpreter, \(\llm\) denotes the LLM decision state, and \(\ahigh\) denotes a high-risk action.

\sysname{} cuts this chain at the lowest-usability-cost point for each carrier class. Static high-authority configuration is protected at write time. Ordinary workspace files are protected at exposed re-entry time through the RTW constraint, enforced dynamically at runtime. Dynamic memory is protected at promotion time through a typed commit protocol. Unavoidable external reads are handled after read time through capability attenuation.

\subsection{Threat Abstraction and Taint Model}

\begin{definition}[Opaque Read]
An opaque read \(\rOpq(f)\) occurs when a trusted local component reads file \(f\) without exposing its raw contents, or an unbounded natural-language rendering of its contents, to the LLM decision context. Opaque reads may return bounded certificates such as error codes, validity results, hashes, or structured diagnostics, but not arbitrary file content.
\end{definition}

\begin{definition}[Exposed Read]
An exposed read \(\rUp(f)\) occurs when the contents of \(f\), or an unbounded natural-language rendering of those contents, are conservatively treated as influencing an authority-bearing interpreter. This includes loading \(f\) into the LLM context, prompt, memory, configuration loader, RAG context, execution engine, workflow engine, or another agent-facing communication channel. We adopt a conservative approximation: any such loading event is classified as an exposed read regardless of whether the content materially alters subsequent decisions, since the latter condition is undecidable in general.
\end{definition}

\begin{definition}[Tainted Carrier]
A carrier \(f\) is tainted at time \(t\) if its current version is attacker-controlled, externally supplied, or derived from a contaminated LLM decision state. Taint propagates conservatively: outputs written by a contaminated LLM state are marked tainted-derived unless they pass an external declassification procedure.
\end{definition}

\begin{definition}[Contaminated LLM State]
The LLM decision state \(\llm\) is conservatively marked contaminated at time \(t\) if there exists an untrusted source \(s\) and a time \(t_r \leq t\) such that an exposed read \(\rUp(s, t_r)\) occurred. The contamination persists until an explicit external declassification procedure clears the label. We adopt this conservative approximation because whether exposed content materially influences subsequent LLM decisions is undecidable in general; the conservative model ensures no tainted influence escapes undetected at the cost of potential over-approximation.
\end{definition}

\begin{definition}[High-Risk Action]
An action \(a\) is high-risk if it can persist, amplify, execute, or propagate tainted influence. This includes writing autoloaded files, trusted memory, configuration, executable scripts, workflows, or shared state; sending messages to other agents or external systems; invoking shell, network, or email tools; modifying policies or permissions; and committing content to cross-session or cross-agent state.
\end{definition}

The central safety invariant is:
\[
\neg \exists f,t_w,t_r,t_a:\;
\wTau(f,t_w)
\land
\rUp(f,t_r)
\land
\ahigh(t_a)
\land
(t_w < t_r < t_a).
\]
That is, attacker-influenced content must not be written into a persistent carrier, later exposed to an authority-bearing decision context, and then used to drive a high-risk action. Theorem~\ref{thm:no-worm} establishes that \sysname{} enforces this invariant under the following assumption.

\begin{assumption}[Conservative Taint Propagation]
\label{asm:taint}
Any output produced by a contaminated LLM decision state \(\contam(\llm)\) is marked tainted-derived and treated as untrusted for all subsequent exposed reads and high-risk actions. Taint is cleared only by an explicit external declassification procedure; the LLM itself cannot declassify its own outputs.
\end{assumption}

This assumption is operationalized by the taint labeling rules in the enforcement mechanism: writes by a contaminated LLM state mark the target carrier as tainted-derived, and tainted-derived carriers are treated as untrusted by the attenuation layer. The assumption is conservative by design---it may over-approximate contamination, but it ensures no tainted influence escapes undetected.

\subsection{Layer 1: RTW for Exposed Re-Entry}

The first layer prevents file-mediated re-entry before tainted content reaches the LLM. The constraint is not imposed on all filesystem reads. It is imposed only on exposed reads \(\rUp\). Opaque reads \(\rOpq\), such as local validation or deterministic parsing, remain available.

Let \(T\) be an operation trace. For a file \(f\), let \(\pi_f(T)\) be the projection of \(T\) onto operations involving \(f\), restricted to the alphabet \(\{R_f^{\uparrow}, W_f\}\), where \(W_f\) denotes a write that may introduce attacker-influenced or contaminated-derived content.

\begin{definition}[RTW-Safe Trace]
A trace \(T\) is RTW-safe for file \(f\) if
\[
\pi_f(T) \in (R_f^{\uparrow})^* W_f^*.
\]
Equivalently, \(T\) is RTW-safe if it contains no subsequence \(W_f \prec R_f^{\uparrow}\).
\end{definition}

\begin{theorem}[No Persistent Worm Propagation]
\label{thm:no-worm}
Under \sysname{} with complete mediation of exposed reads, persistent taint labeling, and capability attenuation, no attacker-controlled content can complete the worm propagation chain
\[
\wTau(f,t_w) \prec \rUp(f,t_r) \prec \ahigh(t_a)
\]
across any sequence of agents or carriers.
\end{theorem}

\begin{proof}
We show that each link of the propagation chain is broken by one layer of the framework.

\textit{Case 1: File-mediated re-entry.}
Suppose tainted content is written to carrier $f$ at time $t_w$ and the agent later attempts an exposed read $\rUp(f, t_r)$ with $t_w < t_r$. Regardless of whether $f$ was written directly or via a multi-carrier chain, any carrier through which tainted content re-enters the LLM must have received a prior tainted write followed by an exposed read: this produces $W_f \prec R_f^{\uparrow}$ in $\pi_f(T)$, violating the RTW constraint. By complete mediation, the runtime denies the exposed read before tainted content reaches the LLM decision context. The chain cannot proceed past $\rUp$.

\textit{Case 2: First-read external contamination.}
Suppose tainted content enters via an unavoidable external read $\rUp(s)$ where no prior write to $s$ exists (RTW does not apply since there is no $W_s \prec R_s^{\uparrow}$ subsequence). The runtime conservatively marks $\contam(\llm)$ upon any exposed read of untrusted content. Capability attenuation then enforces $\contam(\llm) \land \ahigh \Rightarrow \deny$, blocking any high-risk propagation action $\ahigh(t_a)$ that would persist or transmit the tainted influence. The chain cannot proceed past $\ahigh$.

\textit{Combined guarantee.}
Every file-mediated re-entry attempt is eliminated by Case~1. Every propagation action following unavoidable external contamination is blocked by Case~2. The two cases are exhaustive under Assumption~\ref{asm:taint}: any write performed after $\contam(\llm)$ is marked tainted-derived and treated as untrusted, so a subsequent exposed read of that output falls under Case~1 (if a prior write exists) or triggers renewed contamination under Case~2. Therefore no path through $\wTau \prec \rUp \prec \ahigh$ can complete, and the central safety invariant holds. Hence no worm propagation cycle can persist across agents or sessions under \sysname{}.
\end{proof}

Theorem~\ref{thm:no-worm} is the central guarantee of \sysname{}: the three framework components have complementary and non-overlapping scopes. RTW eliminates all file-mediated re-entry, including multi-carrier chains, by enforcing a per-carrier write-before-exposed-read constraint. Capability attenuation covers the residual case where tainted content arrives via unavoidable external reads with no prior file write, by preventing any subsequent high-risk action. Together they cover the complete threat surface without redundancy.

\paragraph{Enforcement.}
RTW is a static trace constraint: it characterizes the set of operation sequences that are free of write-before-exposed-read re-entry. The enforcement rule below provides its dynamic runtime realization by preventing exposed reads on tainted carriers when high-risk capabilities are enabled. Files initialized from a trusted signed baseline begin as \clean. User-provided files, downloaded files, externally synchronized files, and outputs of contaminated decision states begin as \external{} or \tainted. Writes by contaminated decision states mark the target carrier as tainted-derived.

The dynamic enforcement rule implementing RTW is:
\[
\tainted(f)
\land
\rUp(f)
\land
\highcap
\Rightarrow
\deny.
\]
Here \(\highcap\) denotes that the current decision context retains permission to perform high-risk actions without attenuation. This rule dynamically enforces the RTW trace constraint: by denying any exposed read of a tainted carrier while high-risk capabilities remain active, it prevents the runtime from entering any trace that would violate the RTW-safe property. Opaque local reads remain permitted:
\[
\tainted(f)
\land
\rOpq(f)
\Rightarrow
\allow,
\]
provided the local component returns only bounded, sanitized certificates rather than raw content. This preserves workflows such as local validation, parsing, compilation, and format checking without allowing written content to re-enter a high-capability LLM context.

\subsection{Autoloaded Carriers}

Autoloaded carriers require separate handling because the exposed read is system-mediated. If a file is automatically loaded into the agent context, then a write to that file is equivalent to scheduling a future exposed read.

\subsubsection{Static High-Authority Configuration}

Static configuration files define agent identity, policy, or core operational rules. These files are treated as sealed runtime constants:
\[
W(f)=0
\qquad
\text{for } f \in F_c^{static}.
\]
They may be updated only through an offline trusted deployment channel, such as a signed administrative update. At runtime, the agent has no write capability to these files. This removes the agent-mediated attack surface for core identity and security policy files.

\subsubsection{Dynamic Autoloaded State}

Dynamic autoloaded state supports task continuity and long-term adaptation, but it also creates a direct path from current interaction to future LLM context. We separate dynamic carriers into task-local state and long-term memory.

Task-local state is managed through leases. A lease grants bounded write access for a specific task interval and expires automatically:
\[
\lease(f)=[t_0,t_1].
\]
Task-local state has low authority, short lifetime, and cannot encode future tool invocations, permission changes, or cross-session policy.

Long-term memory is not updated through free-form writes. Instead, the runtime separates untrusted candidate memory from trusted autoloaded memory:
\[
M_c \rightarrow M_t.
\]
Automatic summarization, LLM-generated notes, tool outputs, and external content may write only to the candidate store \(M_c\). Entries in \(M_c\) are not autoloaded. A candidate \(c\) is promoted to the trusted memory store \(M_t\) if and only if it satisfies the runtime-defined typed commit policy:
\[
\promote(c)=1
\iff
\schema(c)
\land
\sourceallowed(c)
\land
\scopeallowed(c)
\land
\authority(c)\leq A_{\max}
\land
\ttl(c)\leq T_{\max}.
\]
The biconditional ensures that the policy is both necessary and sufficient: no candidate is promoted without satisfying all conditions, and every candidate satisfying all conditions is eligible for promotion. The LLM may propose memory candidates, but it cannot authorize promotion. The schema registry and promotion policy are defined by the runtime, not by natural-language self-classification. Free-form future instructions, tool permissions, policy updates, cross-user rules, executable commands, and external communication rules fail the \(\schema\) check and are rejected.

When trusted memory is loaded, the runtime renders only a fixed safe projection of \(M_t\) into the LLM context. Raw candidate text, raw evidence, and untrusted summaries are not autoloaded.

\subsection{Layer 2: Capability Attenuation}
\label{sec:attenuation}

RTW prevents file-mediated re-entry from tainted carriers. However, agents must often read unavoidable external sources, including user documents, web pages, emails, and peer-agent messages. These inputs may be tainted on first read, so no prior write-before-read event exists for RTW to block.

For such cases, the second layer applies capability attenuation. After an exposed read of external or untrusted content, the runtime conservatively marks the LLM state as contaminated:
\[
\rUp(s)
\land
\untrusted(s)
\Rightarrow
\contam(\llm).
\]
This layer complements RTW by addressing the case RTW cannot block: tainted content arriving via an unavoidable external source with no prior file write, where no $W_s \prec R_s^{\uparrow}$ subsequence exists for RTW to intercept.

A contaminated LLM state may continue low-risk cognitive operations such as summarization, classification, extraction, and local analysis. Ordinary output files may still be written for usability, but they are marked tainted-derived. High-risk actions are denied or require approval:
\[
\contam(\llm)
\land
\ahigh
\Rightarrow
\deny\ \lor\ \guard.
\]

Capability attenuation is enforced externally by the runtime policy engine, not by the LLM. Thus, the defense does not require the model to correctly distinguish data from instructions. A prompt payload may influence the LLM's proposed action, but it cannot directly grant itself attenuated capabilities under complete mediation of high-risk operations.

Capabilities can be restored only after a context reset or an external declassification procedure. Declassification may include deterministic validation, bounded extraction, typed memory promotion, signed verification, sandboxed execution, or explicit approval. The LLM itself cannot declassify tainted content.

\subsection{Security--Usability Tradeoff}

The framework chooses the least disruptive cut point for each carrier class. For static high-authority configuration, the lowest-cost cut is at write time. For ordinary workspace files, the cut is at exposed re-entry: writes and opaque local reads remain available, but tainted files cannot be exposed to a high-capability LLM. For dynamic memory, the cut is at promotion time: untrusted candidates may be stored, but only typed, scoped, low-authority entries can become autoloaded memory. For unavoidable external reads, the cut is after read time: the read is allowed, but high-risk capabilities are attenuated.

For a carrier \(f\), let \(C_W(f)\), \(C_R(f)\), and \(C_A(f)\) denote the usability cost of cutting the chain at write time, exposed-read time, or high-risk-action time, respectively. The policy goal can be expressed conceptually as selecting a cut set \(\mathcal{C}\) that intersects every dangerous chain while minimizing usability loss:
\[
\min_{\mathcal{C}}
\sum_{c\in\mathcal{C}} C(c)
\quad
\text{s.t.}
\quad
\forall
\bigl(
\wTau(f)
\prec
\rUp(f)
\prec
\ahigh
\bigr),
\quad
\mathcal{C}\cap\{\wTau,\rUp,\ahigh\}\neq\emptyset.
\]
This is a hitting-set formulation and is NP-hard in general; moreover, the set of dangerous chains in an LLM agent system may be infinite and not fully enumerable at design time. We therefore treat this formulation as a conceptual characterization of the security--usability tradeoff rather than a directly solved optimization. In practice, the carrier classification framework instantiates a tractable approximation: each carrier class is assigned a fixed cut point based on its authority and triggerability profile, yielding the layered policy described above.

\subsection{Assumptions and Guarantees}

The guarantees rely on four assumptions. First, high-risk actions and exposed reads are completely mediated by the runtime. Second, carrier labels are persistent across sessions and associated with underlying carrier identity rather than transient path names alone. Third, trusted initialization distinguishes signed or verified clean state from external or user-provided state. Fourth, declassification is performed by mechanisms external to the LLM.

Under these assumptions, the RTW enforcement rule prevents file-mediated re-entry of tainted content into a high-capability decision context, including multi-carrier propagation chains: any such chain must produce a $W_f \prec R_f^{\uparrow}$ violation on the carrier through which tainted content enters the LLM, which RTW directly eliminates. Capability attenuation bounds the residual risk from unavoidable external reads with no prior file write, by preventing contaminated decision states from performing high-risk actions. Together, the two layers cover the complete threat surface and prevent attacker-controlled content from completing the full persistence--re-entry--action chain while preserving ordinary data-processing workflows.
\section{Discussion}

Autonomous propagation in LLM agent ecosystems represents a shift from exploit-driven compromise to semantics-driven compromise. The attacker does not need to exploit a memory-safety bug. The attacker instead exploits the agent's normal mechanism for interpreting natural-language state.

\paragraph{Read as an integrity event.}
In file-backed agents, exposed reads affect integrity because they influence the controller. This does not mean that all reads are dangerous. Opaque local reads by deterministic tools can remain safe and useful. The dangerous operation is exposed re-entry of tainted content into a decision context that retains high-risk capabilities.

\paragraph{Temporal separation.}
The write and the dangerous read need not be adjacent. A tainted carrier may be written during one turn and read much later during a scheduled task, future session, restart, or another agent's interaction. This temporal separation complicates defenses based on per-turn prompt filtering and motivates persistent carrier labels.

\paragraph{Trust-based lateral movement.}
Multi-agent environments create trust-based delegation channels. A low-capability agent can influence a higher-capability agent if the latter processes the former's output as operational context. This resembles lateral movement, but the propagation medium is natural language rather than executable code.

\section{Related Work}
 
\subsection{Prompt Injection and Instruction Hierarchy}
 
Recent work demonstrates that while LLMs generally prioritize system-level instructions, they remain vulnerable to sophisticated prompt injection attacks where malicious content embedded within user inputs can subvert intended behavior~\cite{wallace2024instruction}. Empirical studies show that models exhibit recency and proximity bias: instructions appearing later in the context window, particularly those resembling direct user directives rather than background configuration, receive disproportionate attention and compliance~\cite{liu2024lost}. 
 
Structured queries attempt to separate prompt and data channels~\cite{chen2024struq}, while instruction-hierarchy training teaches models to prioritize privileged instructions over untrusted content~\cite{wallace2024instruction}. These approaches reduce semantic confusion but do not directly address file-backed persistence, delayed re-entry, and scheduled autoloading in autonomous agent environments.
 
\subsection{Self-Replicating Prompts and LLM Worms}
 
Morris II demonstrated that adversarial self-replicating prompts can create worm-like cascades in GenAI-powered applications, establishing the feasibility of autonomous propagation in LLM-based systems~\cite{cohen2024morrisii}. Prompt Infection further showed that malicious instructions can propagate through direct LLM-to-LLM communication in multi-agent systems~\cite{promptinfection2024}. These works establish the foundational threat model but are limited to specific application scenarios, manually constructed attack vectors, and direct message-based propagation without persistence across session boundaries.

Our work extends this line of research to a substantially more capable threat: a fully automated, zero-click, zero-human-interaction worm that propagates across heterogeneous production agent frameworks through file-backed persistent carriers. Unlike prior work, our attack survives agent restarts through autoloaded carrier re-entry, operates without continuous attacker presence, and traverses framework boundaries without platform-specific adaptation.

Concurrent and independent work~\cite{zhang2026clawworm} demonstrates worm feasibility on a specific platform through manually constructed attack vectors. Our work addresses fundamentally different research questions: automated cross-platform carrier discovery without per-platform customization (\sscgv{}), payload robustness under realistic summarization and paraphrasing pipelines (\srpo{}), zero-click autonomous propagation across heterogeneous framework boundaries including 3-hop cross-platform transmission chains, and formal defense guarantees with mathematical proofs. We additionally identify two empirical insights that prior work does not report: that context injection position determines carrier exploitability, and that in LLM-mediated systems read operations can be more dangerous than write operations, inverting traditional integrity threat models.
 
\subsection{Agent Memory Poisoning}
 
Agent memory and RAG-style knowledge stores are increasingly recognized as attack surfaces. AgentPoison studies poisoning of memory or knowledge bases used by LLM agents, demonstrating that corrupted memory can persist and influence future agent behavior~\cite{chen2024agentpoison}. MINJA demonstrates memory injection attacks through query-only interaction without requiring direct access to memory banks~\cite{dong2025minja}.
 
These works motivate our treatment of memory as a dynamic autoloaded carrier that requires special handling in defense frameworks. Our approach differs by separating untrusted candidate memory from trusted typed memory and requiring runtime-mediated promotion rather than allowing free-form memory writes that could introduce persistent contamination.

\section{Limitations}

Our analysis is limited to frameworks for which source code or sufficient runtime instrumentation is available. Closed-source systems may expose similar carrier classes, but \sscgv{} requires source-level or trace-level visibility to locate exposed-read paths automatically.

The attack evaluation uses controlled testbeds and anonymized frameworks. We do not claim that every deployment of the evaluated frameworks is exploitable. Actual risk depends on configuration, messaging-channel exposure, tool permissions, and whether the relevant carriers are writable and later exposed.

The defense framework requires runtime support for carrier labeling, exposed-read mediation, and high-risk action mediation. Systems without a central runtime or with unmanaged direct filesystem access would require additional OS-level or container-level enforcement.

Finally, typed memory promotion limits trusted autoloaded memory to structured low-authority entries. This improves safety but may reduce the expressiveness of free-form memory. We view this as an intentional security--usability tradeoff rather than a complete solution to all memory usability requirements.

\section{Ethics Considerations}

This work studies vulnerabilities in deployed classes of LLM agent frameworks. To reduce risk to users while coordinated disclosure is ongoing, this preprint anonymizes the affected frameworks and omits exploit-enabling details, including exact payload templates, platform-specific carrier paths, exploitation scripts, and version identifiers. We disclosed our findings to the affected maintainers before public release of this preprint and are coordinating timelines for remediation and public attribution.

All experiments were conducted in controlled testbeds using accounts and agents operated by the authors. We did not test against third-party deployments, exfiltrate user data, or interact with public users. The paper reports aggregate results and carrier classes rather than step-by-step exploit instructions. We will release additional implementation details, artifacts, and platform names after mitigations are available or after an agreed coordinated disclosure date.

\section{Conclusion}

This paper presents a systematic study of persistent worm propagation in file-backed LLM agent ecosystems. We introduce \sscgv{} for automated discovery of injectable persistent carriers and \srpo{} for studying payload robustness under LLM-mediated summarization, paraphrasing, and compression. Our controlled evaluation across three anonymized open-source frameworks shows that persistent carriers, exposed reads, and messaging-channel co-presence can enable multi-hop propagation without traditional software exploitation. We also develop a temporal re-entry defense framework. RTW prevents same-carrier write-before-exposed-read re-entry; sealed configuration protects static high-authority files; typed memory promotion prevents free-form summaries from becoming trusted long-term memory; and capability attenuation limits high-risk actions after unavoidable exposed reads. Together, these mechanisms provide a foundation for securing autonomous agent ecosystems while preserving useful data-processing workflows.

\section*{Acknowledgments}

We thank the maintainers and community members who provided feedback during coordinated disclosure. This work was conducted in controlled research environments.

\end{document}